\documentclass[%
 reprint,
 amsmath,amssymb,
 aps,
 pra,
]{revtex4-1}

\usepackage{graphicx}
\usepackage{dcolumn}
\usepackage{bm}

\usepackage[utf8]{inputenc}
\usepackage[english]{babel}
\usepackage[T1]{fontenc}
\usepackage{amsmath}
\usepackage{hyperref}
\usepackage{bm}
\usepackage{dsfont}
\usepackage{amsthm}
\usepackage{verbatim} 
\usepackage{qcircuit}
\usepackage{circuitikz}
\usepackage{caption} 
\usepackage{tikz}
\usepackage{lipsum}
\usepackage[normalem]{ulem}
\usepackage{subcaption}

\newcommand{\bra}[1]{\mbox{$\langle #1|$}}
\newcommand{\ket}[1]{\mbox{$|#1\rangle$}}

\newcommand{\ketbra}[2]{\mbox{$|#1\rangle\langle #2|$}}

\def\eye{\mathds{1}}

\newtheorem{thm}{Theorem}[section]

\newtheorem{prop}[thm]{Proposition}

\newtheorem{prob}{Problem}[section]

\begin{document}

\preprint{APS/123-QED}

\title{Variationally Learning Grover's Quantum Search Algorithm}
\thanks{All supporting data and source code are available online. \cite{GitHub}}%

\author{Mauro E.S.~Morales}
\email{mauricioenrique.moralessoler@skoltech.ru}
\author{Timur Tlyachev}%
\author{Jacob Biamonte}
\email{jacob.biamonte@qubit.org}
 \homepage{DeepQuantum.AI}
\affiliation{Deep Quantum Labs, Skolkovo Institute of Science and Technology, 3 Nobel Street, 
Moscow, Russia 121205}

\date{\today}

\begin{abstract}
Given a parameterized quantum circuit such that a certain setting of these real-valued parameters corresponds to Grover's celebrated search algorithm, can a variational algorithm recover these settings and hence learn Grover's algorithm? We studied several constrained variations of this problem and answered this question in the affirmative, with some caveats. Grover's quantum search algorithm is optimal up to a constant. The success probability of Grover's algorithm goes from unity for two-qubits, decreases for three- and four-qubits and returns near unity for five-qubits then oscillates ever-so-close to unity, reaching unity in the infinite qubit limit. The variationally approach employed here found an experimentally discernible improvement of $5.77\%$ and $3.95\%$ for three- and four-qubits respectively. Our findings are interesting as an extreme example of variational search, and illustrate the promise of using hybrid quantum classical approaches to improve quantum algorithms.   This paper further demonstrates that to find optimal parameters one doesn't need to vary over a family of quantum circuits to find an optimal solution. This result looks promising and points out that there is a set of variational quantum problems with parameters that can be efficiently found on a classical computer for arbitrary number of qubits.
\end{abstract}

\maketitle



Grover's algorithm \cite{Grover-original} is one of the most celebrated quantum algorithms, enabling quantum computers to quadratically outperform classical computers at database search provided database access is restricted to a `black box' -- called the oracle model. In addition to the wide application scope of database search, Grover's algorithm has further applications as a subroutine used in a variety of other quantum algorithms. 

Variational hybrid quantum/classical algorithms have recently become an area of significant interest \cite{original-vqe, Fahri-qaoa1,Nasa-2017,Nasa-qaoa-grover,2017arXiv171001022M,Nasa-qaoa-fermion, 2017arXiv171205304V, 2017arXiv171205771O}.
These algorithms have shown several advantages such as robustness to quantum errors and low coherence time requirements \cite{theory-var}, which make them ideal for implementations in current quantum computer architectures. Here we take inspiration from algorithms such as the variational quantum eigensolver (VQE) \cite{original-vqe} and the quantum approximate optimization algorithm (QAOA) \cite{Fahri-qaoa1}. The general procedure of these variational hybrid quantum/classical algorithms is the following:
\begin{enumerate}
\item Prepare state $\ket{\psi(\theta)}$ using a quantum computer, where $\theta = (\theta_1, \theta_2, ..., \theta_k)$ are tunable parameter(s). The state is prepared by specifying a sequence of $k$ gates $\mathcal{U}^{(1)}(\theta_1) \mathcal{U}^{(2)}(\theta_2)... \mathcal{U}^{(k)}(\theta_k)$ applied to a starting reference state $\ket{r}$. Thus, the prepared state is $\ket{\psi(\theta)} = \mathcal{U}^{(1)}(\theta_1) \mathcal{U}^{(2)}(\theta_2)... \mathcal{U}^{(k)}(\theta_k) \ket{r}$.
\item Measure the expectation value of an objective function $\langle A (\theta) \rangle $. The objective function will depend on the problem to be solved. In the case of VQE the interest is in finding the eigenvalues of a given Hamiltonian, moreover the quantum computer is used to calculate the expectation values of the separate terms in the Hamiltonian. For QAOA the objective function approximates the solution of an optimization problem (for details see \cite{Fahri-qaoa1}). The expectation value of this objective function can be calculated using a quantum computer but can also be efficiently evaluated classically.
\item Using a classical computer and an optimization algorithm find new parameters $\theta'$ that minimize $\langle A (\theta) \rangle $. Having found the new parameters, iterate.
\end{enumerate}

Here we consider a variational approach to the established problem of Grover's search \cite{Grover-original}. Note that Grover's search was generalized to the setting of adiabatic quantum computing in \cite{2002PhRvA..65d2308R, 2002quant.ph..6003V}. Grover's quantum search algorithm has been shown to be asymptotically optimal \cite{Bennet-proof,Boyer-proof,Zalka-proof} and hence provides a limiting test case to apply contemporary variational hybrid quantum/classical algorithms to.

We apply a variational algorithm to see if we can recover Grover's algorithm under several constraining scenarios. We motivate our study by recalling that sequencing two Hermitian projectors (Hamiltonians) can be used to recover Grover's search algorithm exactly. We then constrain the search space.  For example, in one scenario we fix the oracle---as is standard---to apply a phase factor of $-1$ to the marked item when varying the time the diffusion generator is applied. In another scenario, we allow the oracle and the diffusion to take the same angle in all iterations. The main objective is to see if a variational algorithm is capable of recovering Grover's algorithm given different restrictions. A peculiar finding is an experimentally discernible improvement of $5.77\%$ and $3.95\%$ for three- and four-qubits respectively (compared to Grover's search algorithm). 

\section{Problem Statement}
\label{sec:variational-search}

Let $n$ be the number of qubits and let $N=2^n$ be the size of the search space. We are searching for a particular bitstring $\bm{\omega} = \omega_1, \omega_2, \omega_3, ..., \omega_n$. We define a pair of rank-$1$ projectors.
\begin{align}
  \label{def:omega-op}
  P_{\bm{\omega}} &= \ketbra{\bm{\omega}}{\bm{\omega}}\\
  \label{def:plus-op}
  P_{\bm{+}} &= \ketbra{+}{+}^{\otimes n} = \ketbra{\bm{s}}{\bm{s}}
\end{align}
where $\ket{\bm{s}}=\frac{1}{\sqrt{N}}\sum_{\bm{x}\in \left\{0,1\right\}^n} \ket{\bm{x}}$. To find $\ket{\bm{\omega}}$, we consider first an ansatz formed by sequencing operators defined in \eqref{eq:operator-v} and \eqref{eq:operator-k}. These operators prepare a state $\ket{\varphi(\bm{\alpha}, \bm{\beta})}$, defined in \eqref{state_prep}, with vectors $\bm{\alpha} = \alpha_1, \alpha_2, ..., \alpha_p$ and $\bm{\beta}=\beta_1, \beta_2, ..., \beta_p$. We seek to minimize the orthogonal complement of the subspace for the searched string \eqref{eq:projector-w}.
\begin{align}
\label{eq:projector-w}
P_{\bm{\omega}^{\perp}} &= \eye - P_{\bm{\omega}}
\end{align}
We sometimes call \eqref{def:plus-op} the driver Hamiltonian or diffusion operator \cite{Nasa-qaoa-grover}.
The state is varied to minimize this orthogonal component \eqref{eq:minimization}.
\begin{align}
\label{eq:minimization}
\min_{\bm{\alpha},\bm{\beta}} \bra{\varphi(\bm{\alpha}, \bm{\beta})}P_{\bm{\omega}^{\perp}} \ket{\varphi(\bm{\alpha},\bm{\beta})} \geq \min_{\ket{\phi}} \bra{\phi}P_{\bm{\omega}^{\perp}}\ket{\phi}
\end{align}
To prepare the state we develop the sequence \eqref{state_prep}.
\begin{align}
\label{state_prep}
\ket{\varphi(\bm{\alpha},\bm{\beta})} = \mathcal{K}(\beta_p) \mathcal{V}(\alpha_p)\cdots\mathcal{K}(\beta_1)\mathcal{V}(\alpha_1)\ket{\bm{s}}.
\end{align}
Where the operators are defined as \eqref{eq:operator-v} and \eqref{eq:operator-k}.
\begin{align}
\label{eq:operator-v}
\mathcal{V}(\alpha) &:= e^{\imath \alpha P_{\bm\omega}}\\
\label{eq:operator-k}
\mathcal{K}(\beta) &:= e^{\imath \beta P_{\bm{+}}}
\end{align}
The length of the sequence is $2p$, for integer $p$.
We consider the following problems that the variational algorithm will face.

\begin{prob}[Standard Oracle, Variational Diffusion]
\label{def:standar-oracle-variational-diffusion} Find $p$ angles $\bm\beta = (\beta_1, ..., \beta_p)$ and fix $\bm\alpha = (\alpha_1 = \pi, ..., \alpha_p = \pi) $ to minimize \eqref{eq:minimization} via the sequence \eqref{state_prep}, given the operators \eqref{eq:operator-v} and \eqref{eq:operator-k}.
\end{prob}
In this problem, we have fixed the standard black-box oracle of Grover's algorithm and the algorithm optimizes for the angles in the diffusion operator. We also consider a restricted variational problem where all the diffusion operators must apply the same variational angle. 
\begin{prob}[Standard Oracle, Restricted Variational Diffusion]
\label{def:standard-oracle-restricted-variational-diffusion} As in problem \ref{def:standar-oracle-variational-diffusion} except find $p$ angles $\bm\beta = (\beta, ..., \beta)$ and choosing $\bm\alpha = (\alpha_1 = \pi, ..., \alpha_p = \pi)$.
\end{prob}
A third problem to which we will apply the variational algorithm is considering both the oracle and the diffusion angles as variational parameters. We consider in this case a phase matching condition, meaning that angles are restricted so they are equal.
\begin{prob}[Restricted Variational Oracle and Diffusion]
\label{def:restricted-variational-oracle-diffusion} As in \ref{def:standar-oracle-variational-diffusion} except find $2p$ angles $(\bm\alpha, \bm\beta) = (\alpha_1, ..., \alpha_p, \beta_1, ..., \beta_p)$ with the restriction  $\bm\beta = \bm\alpha = \alpha_1, \alpha_2, ..., \alpha_p$ and $\alpha_1 = \alpha_2 = ... = \alpha_p$.
\end{prob}
Finally we consider variations of the oracle angles and the diffusion operator separately.
\begin{prob}[Variational Oracle and Diffusion]
\label{def:variationa-oracle-diffusion} As in problem \ref{def:standar-oracle-variational-diffusion} except find $2p$ angles $(\bm\alpha, \bm\beta) = (\alpha_1, ..., \alpha_p, \beta_1, ..., \beta_p)$.
\end{prob}
Note that the angles obtained in \eqref{eq:minimization} only minimize the selected cost function for a particular number of qubits. Once the number of qubits change, the angles obtained in the minimization do not necessarily give the highest probability to find the searched item. Also its important to note that these angles are independent of $\bm\omega$, if we fix the number of qubits in the problem and run the algorithm with a particular set of angles, then these angles give the same probability no matter the $\bm\omega$ we are looking for. As stated earlier our objective in this work is to see if variational algorithms are able to recover Grover's algorithm, for this we need a way of comparing both algorithms.

To compare these variational algorithms with Grover's algorithm, consider the variational ansatz case for $p=1$. Here we recover Grover's operators as the optimal solution for finding a particular string. To prove this, first note that there is only one pair of angles $(\alpha, \beta)$, so we consider \eqref{eq:operator-v}
and \eqref{eq:operator-k} directly. Since $\ketbra{\bm\omega}{\bm\omega}$ is a projector we can expand \eqref{eq:operator-v}. 
\begin{equation}
\begin{split}
\mathcal{V}(\alpha) &= e^{\imath \alpha \ket{\bm\omega}\bra{\bm\omega}} = \eye + (e^{\imath \alpha} - 1)\ket{\bm\omega}\bra{\bm\omega}
\\
&= \eye - (e^{\imath \widetilde{\alpha}} + 1)\ket{\bm\omega}\bra{\bm\omega}
\end{split}
\end{equation}
Where in the last step we have defined $\widetilde{\alpha} = \alpha - \pi $. Now we expand the unitary for the driver Hamiltonian \eqref{eq:driver}.
\begin{equation}
\begin{split}
\label{eq:driver}
\mathcal{K}(\beta)&=  e^{\imath \beta \ket{\bm s}\bra{\bm s}}\\
&= H^{\otimes n}(\eye + (e^{\imath \beta} - 1)\ket{\bm 0}\bra{\bm 0}) H^{\otimes n}\\
&\sim  H^{\otimes n}(-\eye + (e^{\imath \widetilde{\beta}} + 1)\ket{\bm 0}\bra{\bm 0}) H^{\otimes n}\\
&= (e^{\imath \widetilde{\beta}}+1)\ket{\bm s}\bra{\bm s} - \eye
\end{split}
\end{equation}
Where $\sim$ relates the equivalence class of operators indiscernible by a global phase, $H$ is the Hadamard gate and $\widetilde{\beta}= \beta  - \pi$. Notice that for $\widetilde{\alpha} = \widetilde{\beta} = 0$ Grover's oracle and diffusion operators are recovered.

To see that the variational search includes Grover's operators for the case $p > 1$, let us impose $\alpha_1 = \alpha_2 = ... = \alpha_p$ and $\beta_1 = \beta_2 = ... = \beta_p$. In Fig.~\ref{fig:oracle} and Fig.~\ref{fig:qaoa_circuit} the circuits for the oracle and the diffusion operator respectively are shown.  If $i$ pairs of operators \eqref{eq:operator-v} and \eqref{eq:operator-k} are applied to the initial state $\ket{\bm{s}}$ as in \eqref{state_prep}, then we write the prepared state as \eqref{state2-1}. 

\begin{align}
\label{state2-1}
\ket{\varphi_i} = A_i\frac{1}{\sqrt{N-1}}\sum_{x\neq \bm\omega} \ket{x} + B_i \ket{\bm\omega}
\end{align}

We can relate the amplitudes of one step with the amplitudes of the next step with a recursion such as those that appear in \eqref{recursion1} and \eqref{recursion2}; we express the effect of the operators for the variational search over the state as a matrix \eqref{eq:matrix}.

\begin{align}
\label{eq:matrix}
\begin{pmatrix}
    1+\frac{a(N-1)}{N}    & -a(b+1)\frac{\sqrt{N-1}}{N}  \\
    -a\frac{\sqrt{N-1}}{N}   & (b+1)(1+\frac{a}{N}) 
\end{pmatrix}
\end{align}
Here $a = e^{\imath \alpha} -1$ and $b = e^{\imath \beta} -1$. Notice that for $a = b = -2$ the same relation between amplitudes at different steps in \eqref{recursion1} and \eqref{recursion2} up to a global phase in the definition of the Grover operators is obtained. Thus, the variational search space includes Grover's original algorithm. From this matrix it is also possible to see that if the target state is changed, then the angles found through the variational algorithm will give the same probabilities.

An arbitrary phase applied by the oracle was first proposed in \cite{Zalka:arbitrary-phase} although only remarks regarding the use of an arbitrary phase to get higher probabilities for the searched item were done, afterwards several studies regarding the validity of replacing Grover's oracle and diffusion operator with an arbitrary phase version were made \cite{chinese-article1,chinese-article2,Hoyer:arbitrary-phase,chinese-article3}. The main conclusion is that a phase matching condition is required. This condition roughly stated requires the arbitrary angles in the oracle and the diffusion operator to be approximately equal. To address this, we consider also in this work the problem shown in definition \ref{def:restricted-variational-oracle-diffusion}, restricting the variational angles to be equal.

\section{Results}
\label{sec:comparison}
We have compared the performance of the variational search algorithm to Grover for different number of qubits and for the problems \ref{def:standar-oracle-variational-diffusion}, \ref{def:standard-oracle-restricted-variational-diffusion}, \ref{def:restricted-variational-oracle-diffusion} and \ref{def:variationa-oracle-diffusion}.
Surprisingly in problems \ref{def:standar-oracle-variational-diffusion}, \ref{def:restricted-variational-oracle-diffusion} and \ref{def:variationa-oracle-diffusion} it was found that for a small number of qubits and for the same number of operator applications (or oracle calls) $p$ on which Grover obtains the greatest probability, the variational algorithm achieves greater probabilities for finding the searched string. This advantage for a low number of qubits can be seen in Fig.~\ref{fig:alpha3d}. The same plot is obtained for the variational algorithm defined for problems \ref{def:standar-oracle-variational-diffusion}, \ref{def:restricted-variational-oracle-diffusion} and \ref{def:variationa-oracle-diffusion}. In Fig.~\ref{fig:diff-prob} we show the difference between the maximum probabilities of succesfully finding the string between Grover's algorithm and the variational approach up to $11$ qubits. As the number of qubits grows there are diminishing oscillations in this difference, this agrees with the fact that Grover is asymptotically optimal. This diminishing oscillations go to zero very quickly after $10 qubits$.

In the case of problem \ref{def:standard-oracle-restricted-variational-diffusion} we find through numerics that the advantage over Grover's algorithm is lost. We also show in table \ref{table:comparison}, for the variational problem \ref{def:restricted-variational-oracle-diffusion}, the percentage increase between the variational algorithm and Grover's for the probabilities at the number of oracle calls on which this probability is maximal from $2$ to $6$ qubits. For higher numbers of qubits this difference becomes negligible, although there are small oscillations. The same numbers are obtained (except the angle) for the algorithms in problems \ref{def:standar-oracle-variational-diffusion} and \ref{def:variationa-oracle-diffusion}. We show in Fig.~\ref{fig:alpha3d} the probability for finding the solution as a function of the only angle and number of qubits when considering the algorithm of problem \ref{def:restricted-variational-oracle-diffusion}. In case of problem \ref{def:standar-oracle-variational-diffusion} we recover the same probabilities for the marked state as in Grover's algorithm, without the small improvement. From the matrix in \eqref{eq:matrix} and imposing $a=b$ it is possible to plot the probability as a function of the variational angle $\alpha$ and the number of qubits for the algorithm in problem \ref{def:restricted-variational-oracle-diffusion}. We show this plot for $3$ to $6$ qubits in Fig.~\ref{fig:alpha3d}. 

The local maximum at $\pi$ is clearly seen at $n=3$ qubits but also at this or more qubits this angle is not the global maximum. The variational search manages to find these global maxima by using the basin hopping method \cite{basin} for optimization (the search space of the angle is bounded since we restrict to $\alpha \in \left[0, \pi \right]$) is used. The difference of the probabilities to find the solution between the original Grover's algorithm and the variational search does not diminish monotonically with the number of qubits. The same results are obtained for problems \ref{def:standar-oracle-variational-diffusion}, \ref{def:restricted-variational-oracle-diffusion}, \ref{def:variationa-oracle-diffusion}. In the case of problem \ref{def:standard-oracle-restricted-variational-diffusion}, the difference is negligible. 

For low $N=2^n$, where $n$ is the number of qubits in the search, the variational search provides sequences that are more likely to succeed in finding the solution than Grover's algorithm.

Grover's algorithm has been proven to be optimal. This slight advantage of the variational algorithm over Grover's is possible since the proofs have considered a large number of qubits \citep{Bennet-proof,Boyer-proof,Zalka-proof}. We show in the following that this advantage disappears for large $N$.

To prove this we first consider the following theorem.
\begin{thm}
\label{noadvantage}
The maximum probability achievable for the target state in Grover $\to 1$ as $N \to \infty$.
\end{thm}

Before proving this we give proof of a proposition from Zalka \citep{Zalka-proof}.
\begin{prop}
\label{zelkaprop}
For $N \gg 1$ the probability of measuring the target state $\ket{\bm{\omega}}$ after making $p$ oracle calls in Grover's algorithm is $P_p = \sin^2(p\phi + \phi/2)$. 
\end{prop}
\begin{proof}[Proof of proposition \ref{zelkaprop}]
We call the state of the system for step $i$ of Grover's algorithm \ket{\psi_i} and the target state is denoted $\ket{\bm\omega}$. As an initial state for the algorithm we have $\ket{\psi_0} = \frac{1}{\sqrt{N}}\sum_{x=0}^{N-1} \ket{x}$.
Following \cite{Zalka-proof} we write the state of the system in the $i$th state \eqref{eq:state-i} with \eqref{recursion1} and \eqref{recursion2}.

\begin{equation}
\label{eq:state-i}
\ket{\psi_i} = A_i\frac{1}{\sqrt{N-1}}\sum_{x\neq \bm\omega} \ket{x} + B_i \ket{\bm\omega}
\end{equation}
\begin{align}
\label{recursion1}
A_{i+1} = (1-\frac{2}{N})A_i - 2\frac{\sqrt{N-1}}{N}B_i
\end{align}
\begin{align}
\label{recursion2}
B_{i+1} = 2\frac{\sqrt{N-1}}{N}A_i - (1-\frac{2}{N})B_i
\end{align}

This can be written as the result of applying a rotation \cite{Zalka-proof} with $\cos(\phi) = 1-\frac{2}{N}$ and $\sin(\phi) = 2\frac{\sqrt{N-1}}{N}$.
For the proof of the theorem we are interested in the large $N$ limit. Let us consider then $N \gg 1$, thus $\phi \approx \sin \phi \approx \frac{2}{\sqrt{N}}$. The initial state can be written in terms of this angle.

$$\ket{\psi_0}=\cos(\phi/2)\frac{1}{\sqrt{N-1}}\sum_{x\neq \bm\omega} \ket{x} + \sin(\phi/2) \ket{\bm\omega} $$
After applying a rotation with angle $\phi$, $p$ times (equivalent to applying both operators $p$ times), if $N \gg 1$ then
\begin{align*}
\ket{\psi_i}&=\cos(\phi/2 + p\phi)\frac{1}{\sqrt{N-1}}\sum_{x\neq \bm\omega} \ket{x}\\ &+ \sin(\phi/2 + p \phi) \ket{\bm\omega}
\end{align*}
Hence, the probability of measuring $\ket{\bm\omega}$ after $p$ steps is \eqref{probtarget}.
\begin{align}
\label{probtarget}
P_p = \sin^2(p\phi + \phi/2)
\end{align}
\end{proof}
With this, now we prove theorem \ref{noadvantage}.
\begin{proof}[Proof of theorem \ref{noadvantage}]
From \eqref{probtarget} we can calculate the maximum for the probability in the segment $[0,\pi)$. What we actually need is to calculate $p_{max}$ for which the probability is maximum in said segment. We know there is a maximum at $\pi/2$, thus 
$$p_{max} \phi + \phi/2 = \frac{\pi}{2} \implies p_{max} = \frac{\pi}{2 \phi} - \frac{1}{2} $$
Then, for $N \gg 1$ we have 
\begin{equation}
p_{max} = \frac{\pi \sqrt{N}}{4} - \frac{1}{2}
\end{equation}

Recall that $p$ represents the oracle calls, thus it must be an integer. Then the maximum must be either \eqref{opt1} or \eqref{opt2}.
\begin{equation}
\label{opt1}
p_{max} = \left\lfloor \frac{\pi \sqrt{N}}{4} - \frac{1}{2} \right\rfloor
\end{equation}

\begin{equation}
\label{opt2}
p_{max} = \left\lfloor \frac{\pi \sqrt{N}}{4} - \frac{1}{2} \right\rfloor + 1
\end{equation}
We want to prove that as $N \to \infty$ then the maximum probability goes to one. That is, we want to prove $\lim_{N \to \infty} (p\phi + \phi/2) = \pi/2 $. We prove this with the following limit \eqref{limite} --- we can replace \eqref{opt1} with \eqref{opt2} and the result follows anyway.
\begin{equation}
\label{limite}
\lim_{N \to \infty} \left\lfloor \frac{\pi \sqrt{N}}{4} - \frac{1}{2} \right\rfloor \frac{2}{\sqrt{N}} + \frac{1}{\sqrt{N}} = \frac{\pi}{2}
\end{equation}
Hence, $\sin^2(p\phi + \phi/2) \to 1$.
\end{proof}

\begingroup
\setlength{\tabcolsep}{2.3pt}
\begin{table}
\centering
\begin{tabular}{ c c c c } 
 \hline
 \hline
 $N$ & $100 \times (P_{variational} - P_{Grover})/P_{Grover}$ & step $p_{max}$  & angle\\ 
 \hline
 $2^3$ & 5.77\% & 2 & 2.12$^\text{rad}$  \\ 
 $2^4$ & 3.95\% & 3 & 2.19$^\text{rad}$\\ 
 $2^5$ & 0.08\% & 4 & 2.76$^\text{rad}$ \\
 $2^6$ & 0.34\% & 6 & 2.60$^\text{rad}$ \\
 \hline
 \hline
\end{tabular}
 \caption{Percentage increase between the highest probability for finding the solution after measurement obtained in Grover and the two-level variational ansatz. Percent given as a function of $N = 2^n$ where $n$ is the number of qubits and at step $p_{max}$ on which the probability of finding the solution string is maximum. The same table is obtained for the variational ansatz with one angle or with $2p$ angles. Both the diffusion and oracle use the same angle.}
 \label{table:comparison}
\end{table}
\endgroup

With these results then its clear that the advantage is at best negligible for large $N$.

\section{Physical implementations}

Here we note that an experimental implementation of the algorithm for the search problem for the 3-qubit and 4-qubit case is within reach. Recently there has been an experimental implementation of Grover's algorithm in trapped ions for the $3$ qubit case \cite{3qbit-experiment}. We propose that in such experiments it is possible to implement the variational search algorithm proposed in our work and hence obtain higher probabilities. For such implementations a gate decomposition for circuits shown in Fig. \ref{fig:all_circuits} is needed. The single qubit gates $\mathrm{X}$, $\mathrm{H}$, $\mathrm{R_{\alpha}}$ and $\mathrm{C^{k} NOT}$ (k-controlled $\mathrm{NOT}$ gate with $k=3,4$) are implemented in \cite{3qbit-experiment} for the ion trap system. We just need to show how to decompose the k-controlled phase gate such as those shown in Fig. \ref{fig:all_circuits}. A decomposition is given in \cite{3qbit-experiment} that reduces the 3-controlled $\mathrm{R_{\alpha}}$ gate to a pair of Toffoli gates and a 2-controlled $\mathrm{R_{\alpha}}$, the implementation of which exists experimentally for ions \cite{c-phase-ion}.




\section{Effect of Noise}

We have also compared both the variational and Grover's algorithm in the presence of noise. For this we have used the Forest SDK \cite{rigetti} and considered $T_1$ and $T_2$ noise \cite{nielsen00}. This corresponds to so called longitudinal and transversal relaxation times of qubits in the system. The noise is simulated using a Kraus operator approach, yet the Forest SDK works with wavefunction simulation, thus several stochastic runs are required.  We applied our algorithm using the gate set $CZ$, $R_z(\theta)$, $R_x(\pm \pi/2)$. The Pyquil library considers 1-qubit and 2-qubit noise over $R_x$ and $CZ$ respectively. In Fig.~\ref{fig:qaoa-comp} we show the probabilities achieved for the searched string as a function of $T_1$ and $T_2$ parameters for the 3-qubit search problem. To obtain this plot, the algorithm was run $10000$ times for each $T_1$ and $T_2$ then the probabilities for measuring the searched string were calculated with $T_1$, $T_2$ $\in [10^{-6}  s , 10^{-2}  s]$. The area circled by the red line corresponds to a fit circle inside of which the average average difference between the variational algorithm and the optimal probability of Grover's algorithm is greater than $5\%$. We consider this as a threshold to establish a significant difference even in presence of noise.

\begin{figure}
    \centering
   
        \includegraphics[width=0.5\textwidth]{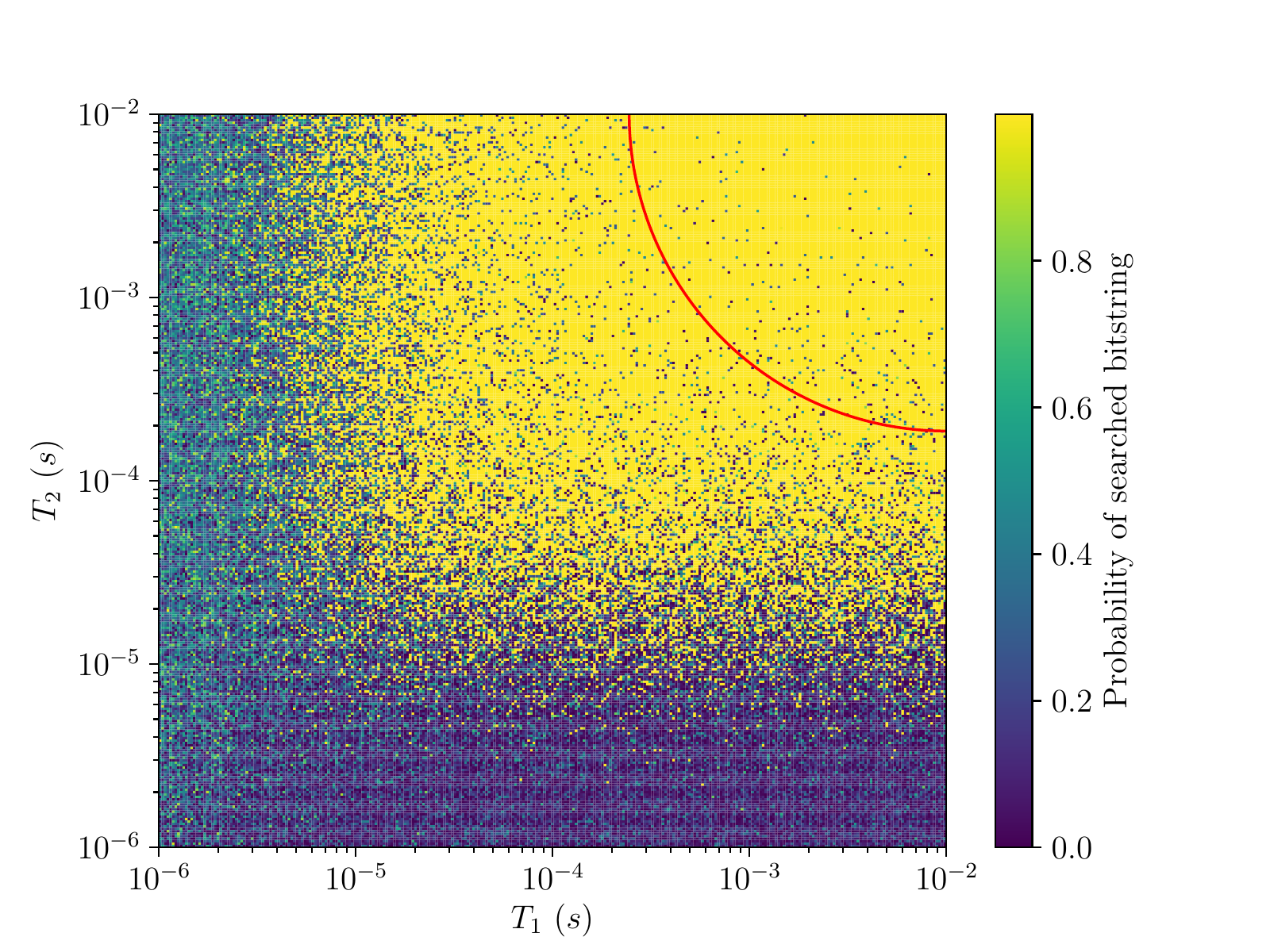}

    \caption{(color online) Probability for finding the solution to the search problem in the case of 3-qubit search under $T_1$ and $T_2$ noise using variational search. The red circumference in the plot corresponds to a fit circle defined such that the probabilities inside the circle are, on average, larger than $5\%$ of the optimal probability obtained with Grover's algorithm.}
    \label{fig:qaoa-comp}

\end{figure}

\section{Discussion}

Let us reflect on several features relating to the most promising results of this study and contrast these with some of the more peculiar shortcomings.

An unusual feature of Grover was the small oscillation in the success probability going form unity for two-qubits and then decreasing for three- and four-qubits. This provided some slack for our variational approach to remove.  In fact, with one angle (shared by both the oracle as well as the diffusion operator), we don't always match the performance of Grover (for five or more qubits). Additional angles however add more degrees of freedom for the optimization procedure to succeed. When considering variational angles for the diffusion operator but not for the oracle operator, theres is still the same advantage as in the case of the variational algorithm with only one angle shared between oracle and diffusion operator. The advantage disappears when only one variational angle is considered for the diffusion operator and the standard oracle of Grover's algorithm is used.

We have presented a transfer matrix \eqref{eq:matrix}, interestingly, the optimization procedure is general in the sense that if we restrict to this transfer matrix, this defines the angle(s) for any search item. In other words, if we find the corresponding angle(s) for a given number of items to search in, we can use this same angle again and again for different search items. Nonetheless it must be noted that the angle(s) obtained through the variational algorithm will only give optimal probabilities for a given number of qubits $n$ in the search problem. For a different $n$, the algorithm needs to be rerun.

One of the main problems of variational algorithms is how to do classical preprocessing, determine the optimal parameters of an algorithm and do it efficiently. This paper demonstrates how this procedure can be done for every finite dimensional search-space on a classical computer. It means that to find optimal parameters one doesn't need to vary over a family of quantum circuits during experiment to find an optimal solution. This result looks promising and points out that there is a set of variational quantum problems with parameters that can be efficiently found on a classical computer for arbitrary number of qubits.

Around the same time as we produced these results, two other teams put out papers which apply machine learning to the optimization and discovery of quantum algorithms \cite{2018arXiv180304114C,2018arXiv180610448W}.  These works---appearing just before and then just after ours---are complimentary as their approach and the algorithms considered are rather different. In \cite{2018arXiv180304114C} the Swap test algorithm for state overlap is learned using a gate sequencing method similar to ours (but different gate set), surprisingly they find shorter gate depths with the variational algorithm compared to the Swap test. In the case of \cite{2018arXiv180610448W}, the authors try a machine learning approach on Simon's algorithm. The optimization method utilized is gradient descent assisted by genetic algorithms by varying over unitaries that depend on a given parameter, they find that this method finds circuits with the same performance as Simon's algorithm. Nonetheless, our work and these two recent papers do share the general concept of training circuits for known algorithms and the results in these studies illustrate that in the future quantum algorithm design might be deeply tied with the methods presented here.

\section{Acknowledgements}
We thank Vladimir Korepin for providing feedback.  Source code producing the plots and generating the numerical findings in this study can be found at {\it Deep Quantum Labs} Github page \cite{GitHub}.

\bibliographystyle{naturemag}  
\onecolumngrid
\bibliography{bibqaoa}

\begin{thebibliography}{10}
\expandafter\ifx\csname url\endcsname\relax
  \def\url#1{\texttt{#1}}\fi
\expandafter\ifx\csname urlprefix\endcsname\relax\def\urlprefix{URL }\fi
\providecommand{\bibinfo}[2]{#2}
\providecommand{\eprint}[2][]{\url{#2}}

\bibitem{GitHub}
\bibinfo{author}{DeepQuantum}.
\newblock \bibinfo{title}{Variationally learning grover}.
\newblock
  \bibinfo{howpublished}{\url{https://github.com/Quantum-Machine-Learning-Initiative/Variationally_learning_grover}}
  (\bibinfo{year}{2018}).

\bibitem{Grover-original}
\bibinfo{author}{{Grover}, L.~K.}
\newblock \bibinfo{title}{{A fast quantum mechanical algorithm for database
  search}}.
\newblock \emph{\bibinfo{journal}{eprint arXiv:quant-ph/9605043}}
  (\bibinfo{year}{1996}).
\newblock \eprint{quant-ph/9605043}.

\bibitem{original-vqe}
\bibinfo{author}{{Peruzzo}, A.} \emph{et~al.}
\newblock \bibinfo{title}{{A variational eigenvalue solver on a photonic
  quantum processor}}.
\newblock \emph{\bibinfo{journal}{Nature Communications}}
  \textbf{\bibinfo{volume}{5}}, \bibinfo{pages}{4213} (\bibinfo{year}{2014}).
\newblock \eprint{1304.3061}.

\bibitem{Fahri-qaoa1}
\bibinfo{author}{{Farhi}, E.}, \bibinfo{author}{{Goldstone}, J.} \&
  \bibinfo{author}{{Gutmann}, S.}
\newblock \bibinfo{title}{{A Quantum Approximate Optimization Algorithm}}.
\newblock \emph{\bibinfo{journal}{ArXiv e-prints}}  (\bibinfo{year}{2014}).
\newblock \eprint{1411.4028}.

\bibitem{Nasa-2017}
\bibinfo{author}{{Hadfield}, S.} \emph{et~al.}
\newblock \bibinfo{title}{{From the Quantum Approximate Optimization Algorithm
  to a Quantum Alternating Operator Ansatz}}.
\newblock \emph{\bibinfo{journal}{ArXiv e-prints}}  (\bibinfo{year}{2017}).
\newblock \eprint{1709.03489}.

\bibitem{Nasa-qaoa-grover}
\bibinfo{author}{{Jiang}, Z.}, \bibinfo{author}{{Rieffel}, E.~G.} \&
  \bibinfo{author}{{Wang}, Z.}
\newblock \bibinfo{title}{{Near-optimal quantum circuit for Grover's
  unstructured search using a transverse field}}.
\newblock \emph{\bibinfo{journal}{PRA}} \textbf{\bibinfo{volume}{95}},
  \bibinfo{pages}{062317} (\bibinfo{year}{2017}).
\newblock \eprint{1702.02577}.

\bibitem{2017arXiv171001022M}
\bibinfo{author}{{Moll}, N.} \emph{et~al.}
\newblock \bibinfo{title}{{Quantum optimization using variational algorithms on
  near-term quantum devices}}.
\newblock \emph{\bibinfo{journal}{ArXiv e-prints}}  (\bibinfo{year}{2017}).
\newblock \eprint{1710.01022}.

\bibitem{Nasa-qaoa-fermion}
\bibinfo{author}{{Wang}, Z.}, \bibinfo{author}{{Hadfield}, S.},
  \bibinfo{author}{{Jiang}, Z.} \& \bibinfo{author}{{Rieffel}, E.~G.}
\newblock \bibinfo{title}{{Quantum approximate optimization algorithm for
  MaxCut: A fermionic view}}.
\newblock \emph{\bibinfo{journal}{PRA}} \textbf{\bibinfo{volume}{97}},
  \bibinfo{pages}{022304} (\bibinfo{year}{2018}).
\newblock \eprint{1706.02998}.

\bibitem{2017arXiv171205304V}
\bibinfo{author}{{Verdon}, G.}, \bibinfo{author}{{Broughton}, M.} \&
  \bibinfo{author}{{Biamonte}, J.}
\newblock \bibinfo{title}{{A quantum algorithm to train neural networks using
  low-depth circuits}}.
\newblock \emph{\bibinfo{journal}{ArXiv e-prints}}  (\bibinfo{year}{2017}).
\newblock \eprint{1712.05304}.

\bibitem{2017arXiv171205771O}
\bibinfo{author}{{Otterbach}, J.~S.} \emph{et~al.}
\newblock \bibinfo{title}{{Unsupervised Machine Learning on a Hybrid Quantum
  Computer}}.
\newblock \emph{\bibinfo{journal}{ArXiv e-prints}}  (\bibinfo{year}{2017}).
\newblock \eprint{1712.05771}.

\bibitem{theory-var}
\bibinfo{author}{McClean, J.~R.}, \bibinfo{author}{Romero, J.},
  \bibinfo{author}{Babbush, R.} \& \bibinfo{author}{Aspuru-Guzik, A.}
\newblock \bibinfo{title}{The theory of variational hybrid quantum-classical
  algorithms}.
\newblock \emph{\bibinfo{journal}{New Journal of Physics}}
  \textbf{\bibinfo{volume}{18}}, \bibinfo{pages}{023023}
  (\bibinfo{year}{2016}).
\newblock \urlprefix\url{http://stacks.iop.org/1367-2630/18/i=2/a=023023}.

\bibitem{2002PhRvA..65d2308R}
\bibinfo{author}{{Roland}, J.} \& \bibinfo{author}{{Cerf}, N.~J.}
\newblock \bibinfo{title}{{Quantum search by local adiabatic evolution}}.
\newblock \emph{\bibinfo{journal}{Physical Review A}}
  \textbf{\bibinfo{volume}{65}}, \bibinfo{pages}{042308}
  (\bibinfo{year}{2002}).
\newblock \eprint{quant-ph/0107015}.

\bibitem{2002quant.ph..6003V}
\bibinfo{author}{van Dam, W.}, \bibinfo{author}{Mosca, M.} \&
  \bibinfo{author}{Vazirani, U.}
\newblock \bibinfo{title}{How powerful is adiabatic quantum computation?}
\newblock In \emph{\bibinfo{booktitle}{Proceedings 2001 IEEE International
  Conference on Cluster Computing}}, \bibinfo{pages}{279--287}
  (\bibinfo{year}{2001}).

\bibitem{Bennet-proof}
\bibinfo{author}{Bennett, C.~H.}, \bibinfo{author}{Bernstein, E.},
  \bibinfo{author}{Brassard, G.} \& \bibinfo{author}{Vazirani, U.}
\newblock \bibinfo{title}{Strengths and weaknesses of quantum computing}.
\newblock \emph{\bibinfo{journal}{SIAM Journal on Computing}}
  \textbf{\bibinfo{volume}{26}}, \bibinfo{pages}{1510--1523}
  (\bibinfo{year}{1997}).

\bibitem{Boyer-proof}
\bibinfo{author}{{Boyer}, M.}, \bibinfo{author}{{Brassard}, G.},
  \bibinfo{author}{{H{\o}yer}, P.} \& \bibinfo{author}{{Tapp}, A.}
\newblock \bibinfo{title}{{Tight Bounds on Quantum Searching}}.
\newblock \emph{\bibinfo{journal}{Fortschritte der Physik}}
  \textbf{\bibinfo{volume}{46}}, \bibinfo{pages}{493--505}
  (\bibinfo{year}{1998}).
\newblock \eprint{quant-ph/9605034}.

\bibitem{Zalka-proof}
\bibinfo{author}{Zalka, C.}
\newblock \bibinfo{title}{Grover's quantum searching algorithm is optimal}.
\newblock \emph{\bibinfo{journal}{Phys. Rev. A}} \textbf{\bibinfo{volume}{60}},
  \bibinfo{pages}{2746--2751} (\bibinfo{year}{1999}).

\bibitem{Zalka:arbitrary-phase}
\bibinfo{author}{{Zalka}, C.}
\newblock \bibinfo{title}{{A Grover-based quantum search of optimal order for
  an unknown number of marked elements}}.
\newblock \emph{\bibinfo{journal}{eprint arXiv:quant-ph/9902049}}
  (\bibinfo{year}{1999}).
\newblock \eprint{quant-ph/9902049}.

\bibitem{chinese-article1}
\bibinfo{author}{GuiLu, L.}, \bibinfo{author}{WeiLin, Z.},
  \bibinfo{author}{YanSong, L.} \& \bibinfo{author}{Li, N.}
\newblock \bibinfo{title}{Arbitrary phase rotation of the marked state cannot
  be used for grover's quantum search algorithm}.
\newblock \emph{\bibinfo{journal}{Communications in Theoretical Physics}}
  \textbf{\bibinfo{volume}{32}}, \bibinfo{pages}{335} (\bibinfo{year}{1999}).

\bibitem{chinese-article2}
\bibinfo{author}{Li, X.}, \bibinfo{author}{Song, K.}, \bibinfo{author}{Sun, N.}
  \& \bibinfo{author}{Zhao, C.}
\newblock \bibinfo{title}{Phase matching in grover's algorithm}.
\newblock In \emph{\bibinfo{booktitle}{Proceedings of the 32nd Chinese Control
  Conference}}, \bibinfo{pages}{7939--7942} (\bibinfo{year}{2013}).

\bibitem{Hoyer:arbitrary-phase}
\bibinfo{author}{H\o{}yer, P.}
\newblock \bibinfo{title}{Arbitrary phases in quantum amplitude amplification}.
\newblock \emph{\bibinfo{journal}{Phys. Rev. A}} \textbf{\bibinfo{volume}{62}},
  \bibinfo{pages}{052304} (\bibinfo{year}{2000}).

\bibitem{chinese-article3}
\bibinfo{author}{Li, C.-M.}, \bibinfo{author}{Hwang, C.-C.},
  \bibinfo{author}{Hsieh, J.-Y.} \& \bibinfo{author}{Wang, K.-S.}
\newblock \bibinfo{title}{General phase-matching condition for a quantum
  searching algorithm}.
\newblock \emph{\bibinfo{journal}{Phys. Rev. A}} \textbf{\bibinfo{volume}{65}},
  \bibinfo{pages}{034305} (\bibinfo{year}{2002}).

\bibitem{basin}
\bibinfo{author}{Wales, D.~J.} \& \bibinfo{author}{Doye, J. P.~K.}
\newblock \bibinfo{title}{Global optimization by basin-hopping and the lowest
  energy structures of lennard-jones clusters containing up to 110 atoms}.
\newblock \emph{\bibinfo{journal}{The Journal of Physical Chemistry A}}
  \textbf{\bibinfo{volume}{101}}, \bibinfo{pages}{5111--5116}
  (\bibinfo{year}{1997}).

\bibitem{3qbit-experiment}
\bibinfo{author}{{Figgatt}, C.} \emph{et~al.}
\newblock \bibinfo{title}{{Complete 3-Qubit Grover search on a programmable
  quantum computer}}.
\newblock \emph{\bibinfo{journal}{Nature Communications}}
  \textbf{\bibinfo{volume}{8}}, \bibinfo{pages}{1918} (\bibinfo{year}{2017}).
\newblock \eprint{1703.10535}.

\bibitem{c-phase-ion}
\bibinfo{author}{Malinovsky, V.~S.}, \bibinfo{author}{Sola, I.~R.} \&
  \bibinfo{author}{Vala, J.}
\newblock \bibinfo{title}{Phase-controlled two-qubit quantum gates}.
\newblock \emph{\bibinfo{journal}{Phys. Rev. A}} \textbf{\bibinfo{volume}{89}},
  \bibinfo{pages}{032301} (\bibinfo{year}{2014}).
\newblock \urlprefix\url{https://link.aps.org/doi/10.1103/PhysRevA.89.032301}.

\bibitem{rigetti}
\bibinfo{author}{Smith, R.~S.}, \bibinfo{author}{Curtis, M.~J.} \&
  \bibinfo{author}{Zeng, W.~J.}
\newblock \bibinfo{title}{A practical quantum instruction set architecture}
  (\bibinfo{year}{2016}).

\bibitem{nielsen00}
\bibinfo{author}{Nielsen, M.~A.} \& \bibinfo{author}{Chuang, I.~L.}
\newblock \emph{\bibinfo{title}{Quantum Computation and Quantum Information}}
  (\bibinfo{publisher}{Cambridge University Press}, \bibinfo{year}{2000}).

\bibitem{2018arXiv180304114C}
\bibinfo{author}{{Cincio}, L.}, \bibinfo{author}{{Suba{\c s}{\i}}, Y.},
  \bibinfo{author}{{Sornborger}, A.~T.} \& \bibinfo{author}{{Coles}, P.~J.}
\newblock \bibinfo{title}{{Learning the quantum algorithm for state overlap}}.
\newblock \emph{\bibinfo{journal}{ArXiv e-prints}}  (\bibinfo{year}{2018}).
\newblock \eprint{1803.04114}.

\bibitem{2018arXiv180610448W}
\bibinfo{author}{{Wan}, K.~H.}, \bibinfo{author}{{Liu}, F.},
  \bibinfo{author}{{Dahlsten}, O.} \& \bibinfo{author}{{Kim}, M.~S.}
\newblock \bibinfo{title}{{Learning Simon's quantum algorithm}}.
\newblock \emph{\bibinfo{journal}{ArXiv e-prints}}  (\bibinfo{year}{2018}).
\newblock \eprint{1806.10448}.

\bibitem{Barenco95}
\bibinfo{author}{{Barenco}, A.} \emph{et~al.}
\newblock \bibinfo{title}{{Elementary gates for quantum computation}}.
\newblock \emph{\bibinfo{journal}{Physical Review A}}
  \textbf{\bibinfo{volume}{52}}, \bibinfo{pages}{3457--3467}
  (\bibinfo{year}{1995}).
\newblock \eprint{quant-ph/9503016}.

\end{thebibliography}

\newpage
\pagebreak
\clearpage

\onecolumngrid

\begin{figure}[t]
\begin{subfigure}[t]{0.49\textwidth}
        \includegraphics[width=\textwidth]{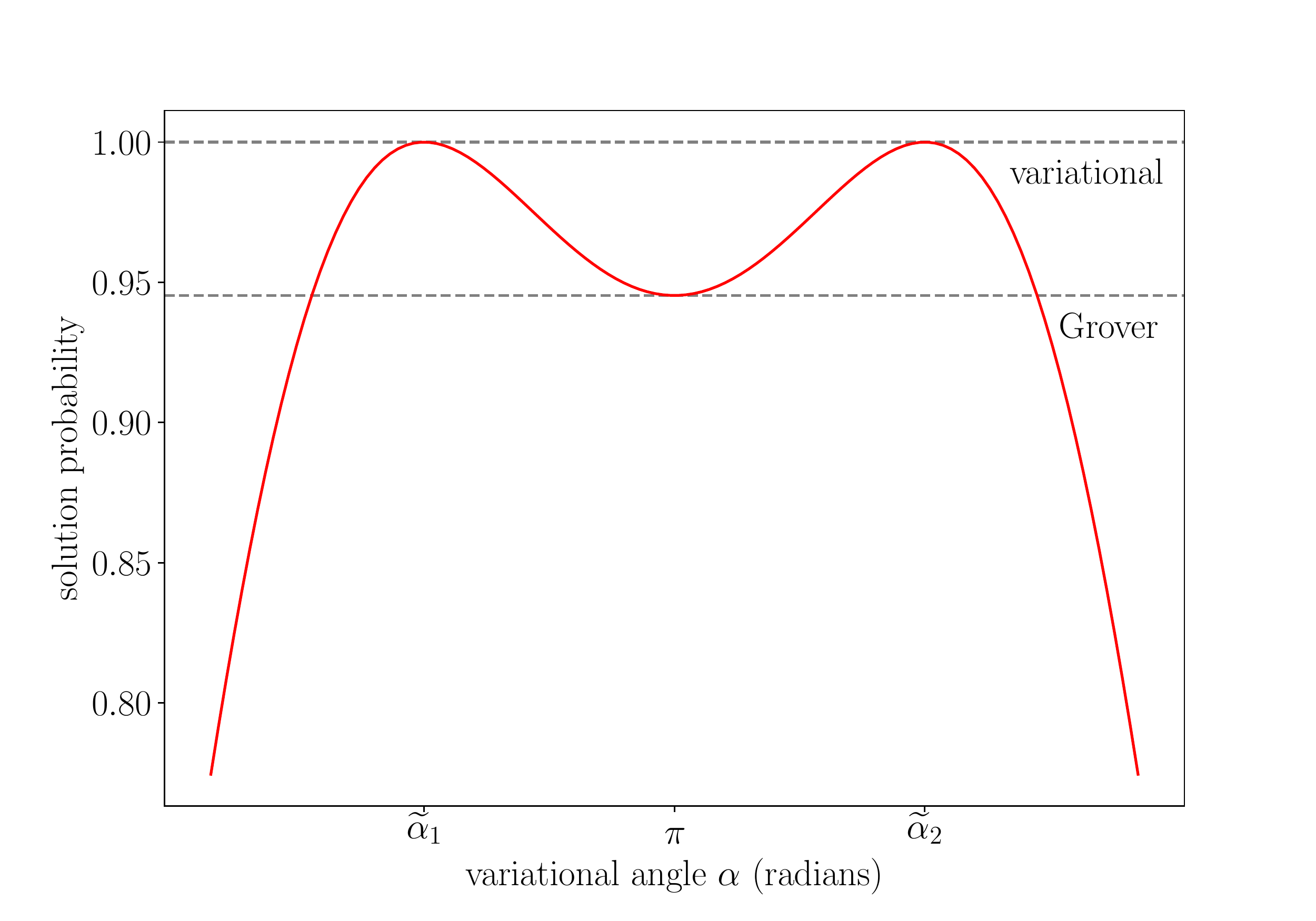}
        \caption{}
    \end{subfigure}
    \begin{subfigure}[t]{0.49\textwidth}              
        \includegraphics[width=\textwidth]{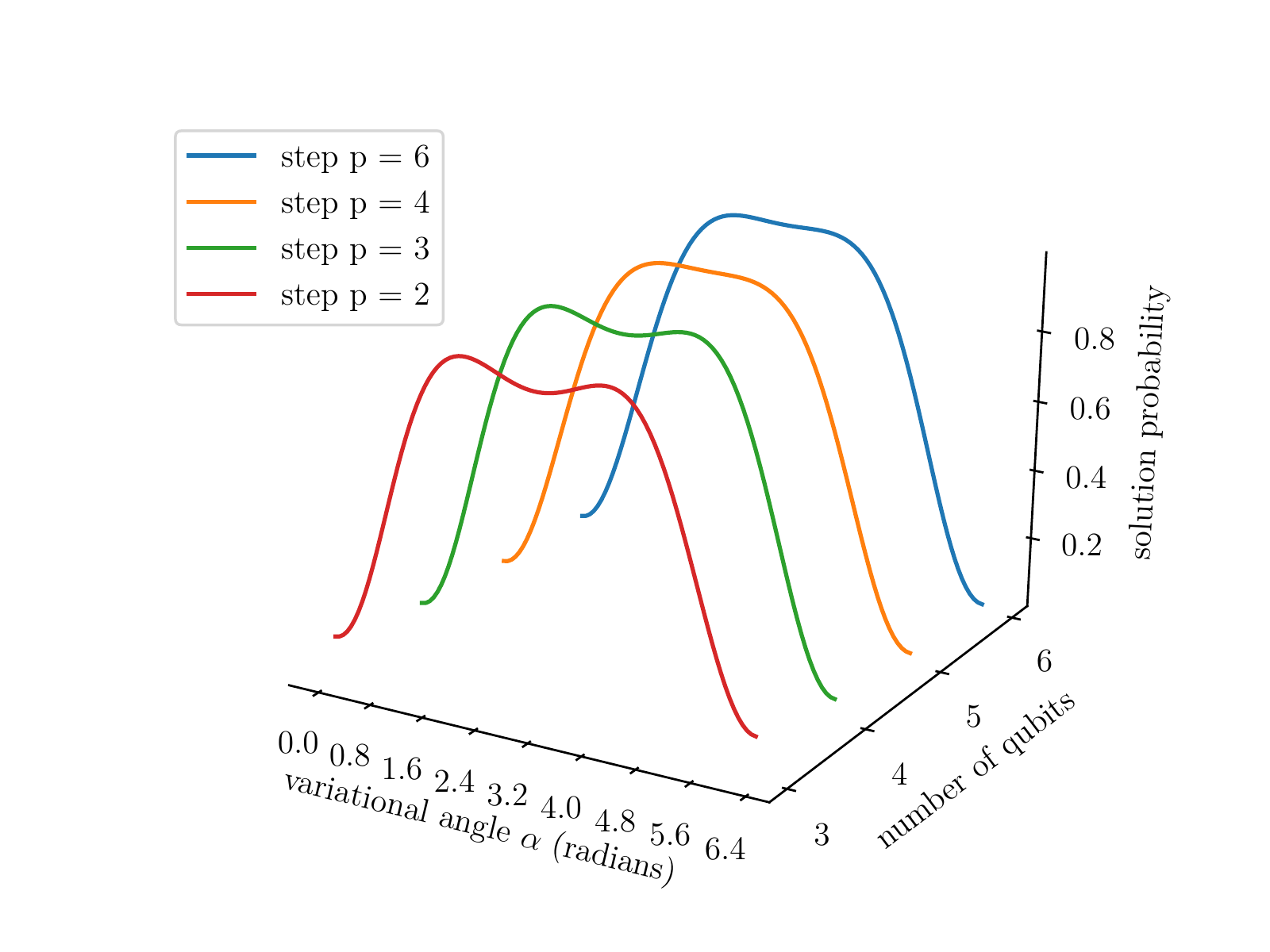}
        \caption{}
    \end{subfigure}
   
    \caption{(color online) (left) Grover's algorithm takes a saddle point between two hills.  Variational search recovers the hill peaks.  Note that the valley becomes increasingly less pronounced past four qubits, providing negligible range for improvement. (right) Probability as a function of the variational angle for the 3 qubit case. Grover's algorithm is recovered in the case $\alpha = \pi$, the variational algorithm obtains angles $\widetilde{\alpha}_{1} = 2.12^{rad}$  and $\widetilde{\alpha}_2 = 2\pi - \widetilde{\alpha}_{1}$  }
    \label{fig:alpha3d}
\end{figure}

\begin{figure}
    \centering
   
        \includegraphics[width=0.7\textwidth]{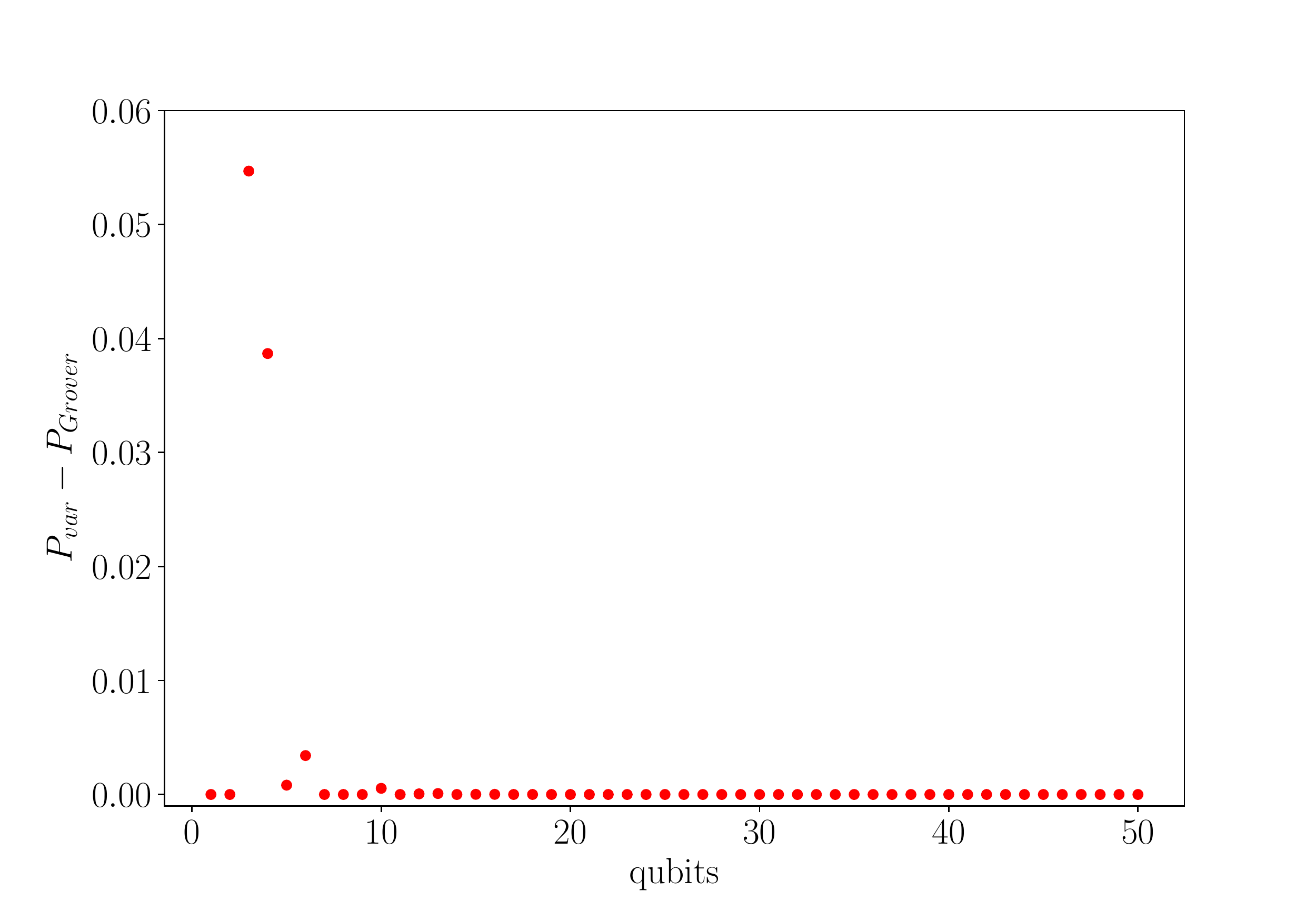}
        
    

    \caption{Difference between variational and Grover. As the number of qubits grows there are exponentially diminishing oscillations in this difference. Each probability is defined for the optimal step in Grover's algorithm for the number of qubits.}
    \label{fig:diff-prob}
\end{figure}

       

        

       

\begin{figure}[htb]

\begin{subfigure}[b]{0.5\textwidth}
\mbox{ 
\Qcircuit @C=2.5em @R=3em {
& \gate{X^{1-\omega_1}} & \ctrl{3} & \gate{X^{1-\omega_1}} & \qw \\
& \gate{X^{1-\omega_2}} & \control \qw & \gate{X^{1-\omega_2}} & \qw \\
& \vdots &  & \vdots &  \\
& \gate{X^{1-\omega_n}} & \gate{R_{\alpha}} $e^{i\alpha}$ \qw & \gate{X^{1-\omega_n}} & \qw
}
}
\caption{oracle circuit}
\label{fig:oracle}
\end{subfigure}
\begin{subfigure}[b]{0.49\textwidth}
\mbox{ 
\Qcircuit @C=2.5em @R=3em {
& \gate{H} &\gate{X} & \ctrl{3} & \gate{X} &  \gate{H} & \qw \\
& \gate{H} &\gate{X} & \ctrl{2} \qw & \gate{X} & \gate{H} & \qw \\
& \vdots & &  & \vdots &  \\
& \gate{H} &\gate{X} & \gate{R_{\beta}} $e^{i\alpha}$ \qw & \gate{X} & \gate{H} & \qw
}
}
\caption{diffusion circuit}
\label{fig:qaoa_circuit}
\end{subfigure}

\caption{Circuit realization of diffusion and oracle circuits. Oracle and diffusion operators can be rewritten via $n$-body control gates $\mathcal{V}(\alpha) = \bigotimes_{i=1}^{n}X_i^{1-\omega_i}\biggr(\eye_{2^n-1}\oplus e^{\imath\alpha}\biggl)\bigotimes_{i=1}^{n}X_i^{1-\omega_i}$ and
$\mathcal{K}(\beta) = H^{\otimes n}X^{\otimes n}\biggr(\eye_{2^n-1} \oplus e^{\imath\beta} \biggl)X^{\otimes n}H^{\otimes n}$ and therefore can be realized using $O(n^2)$ basic gates \cite{Barenco95},
 here operator $\eye_{2^n-1}$ is the $(2^n-1)\times(2^n-1)$ identity matrix. (See also the gate realizations in \cite{3qbit-experiment} which can be readily bootstrapped to realize (a) and (b) above). 
}
\label{fig:all_circuits}
\end{figure}
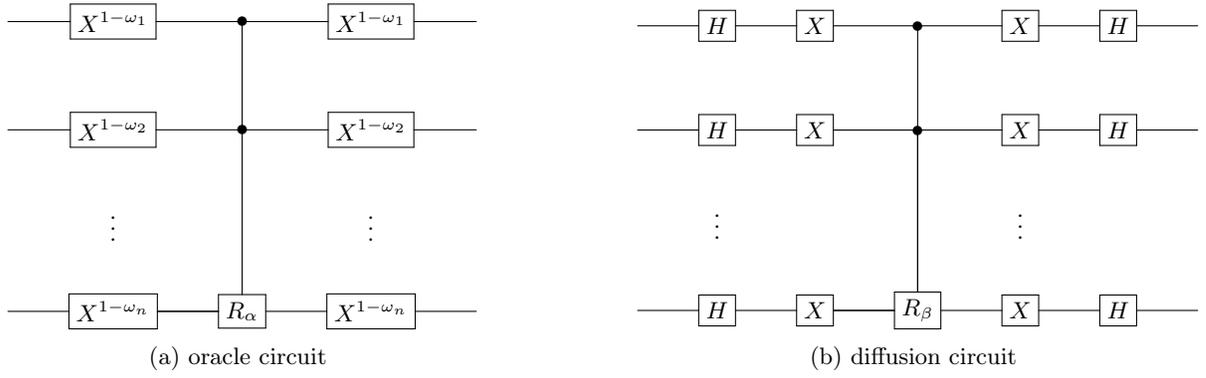




\end{document}